%% file: main.tex
\documentclass[conference]{IEEEtran}
\IEEEoverridecommandlockouts
\usepackage{cite}
\usepackage{amsmath,amssymb,amsfonts}
\usepackage[lined,ruled,linesnumbered,noend]{algorithm2e}
\usepackage{multirow}
\usepackage{url}
\usepackage{tikz} 
\usepackage{caption}
\usepackage{subcaption}
\usepackage{dsfont}
\usepackage{enumitem}
\usepackage{wrapfig}
\usepackage{balance}

\usepackage{graphicx}
\usepackage{textcomp}
\usepackage{xcolor}
\usepackage{amsthm}
\def\BibTeX{{\rm B\kern-.05em{\sc i\kern-.025em b}\kern-.08em
    T\kern-.1667em\lower.7ex\hbox{E}\kern-.125emX}}

\begin{document}

%\title{Securing Pool Selection in Federated Learning:  Attack and Defense Strategies}
\title{Blockchain-based Secure Client Selection in Federated Learning}

% \and
% \and
% \thanks{The first two authors contribute equally to this paper.}
% }

\author{
\IEEEauthorblockN{Truc Nguyen\IEEEauthorrefmark{1} --  Phuc Thai\IEEEauthorrefmark{2}, Tre' R. Jeter\IEEEauthorrefmark{1}, Thang N. Dinh\IEEEauthorrefmark{2} and My T. Thai\IEEEauthorrefmark{1}}

\IEEEauthorblockA{\IEEEauthorrefmark{1}Department of Computer \& Information Science \& Engineering, University of Florida, Gainesville, FL, 32611}
\IEEEauthorblockA{\IEEEauthorrefmark{2}Department of Computer Science, Virginia Commonwealth University, Richmond, VA, 23284\\
Email: truc.nguyen@ufl.edu, thaipd@vcu.edu, t.jeter@ufl.edu, tndinh@vcu.edu, mythai@cise.ufl.edu}
\thanks{The first two authors contribute equally to this paper.}
}

\IEEEoverridecommandlockouts
\IEEEpubid{\makebox[\columnwidth]{978-1-6654-9538-7/22/\$31.00~\copyright2022 IEEE \hfill} \hspace{\columnsep}\makebox[\columnwidth]{ }}

\maketitle

\IEEEoverridecommandlockouts
\IEEEpubidadjcol

\input{headers}
\begin{abstract}
Despite the great potential of Federated Learning (FL) in large-scale distributed learning, the current system is still subject to several privacy issues due to the fact that local models trained by clients are exposed to the central server. Consequently, secure aggregation protocols for FL have been developed to conceal the local models from the server. However, we show that, by manipulating the client selection process, the server can circumvent the secure aggregation to learn the local models of a victim client, indicating that secure aggregation alone is inadequate for privacy protection. To tackle this issue, we leverage blockchain technology to propose a verifiable client selection protocol. Owing to the immutability and transparency of blockchain, our proposed protocol enforces a random selection of clients, making the server unable to control the selection process at its discretion. We present security proofs showing that our protocol is secure against this attack. Additionally, we conduct several experiments on an Ethereum-like blockchain to demonstrate the feasibility and practicality of our solution.
\end{abstract}
\input{intro}

\input{model}
\input{federated-learning}

\input{blockchain}
\input{attack}
\input{problem}
\input{client-selection}
%\input{communication}

\input{analysis}

\input{performance}
\input{related}
\input{conclude}

\section*{Acknowledgement}
This work was supported in part by the National Science Foundation under grants CNS-2140477 and CNS-2140411.
%\newpage
\bibliographystyle{plain}
\balance
\bibliography{FL}
%\appendix
%\input{appendix}

\end{document}

%% file: headers.tex
\newtheorem{theorem}{Theorem}[section]
\newtheorem{definition}[theorem]{Definition}
\newtheorem{claim}[theorem]{Claim}
\newtheorem{remark}[theorem]{Remark}
\newtheorem{lemma}[theorem]{Lemma}
\newtheorem{corol}[theorem]{Corollary}

\newcommand{\pnote}[1]{{\color{blue}{\bf [Phuc: #1]\\}}}

\newcommand{\mybox}[5]{
	\begin{figure}[htbp!]%[htb!]%[tp]
		\small
		\centering
		\begin{tikzpicture}
		\node[anchor=text,text width=\columnwidth-.4cm, draw, rounded corners, line width=1pt, fill=#3, inner sep=2mm] (big) {\\#4};
		\node[draw, rounded corners, line width=.5pt, fill=#2, anchor=west, xshift=5mm] (small) at (big.north west) {#1};
		\end{tikzpicture}
		\caption{#5}
	\end{figure}
}

\newcommand{\VRFgen}{\ensuremath{\mathsf{VRFgen}}}
\newcommand{\VRFprove}{\ensuremath{\mathsf{VRFprove}}}
\newcommand{\VRFverify}{\ensuremath{\mathsf{VRFverify}}}
\newcommand{\sk}{\ensuremath{\mathsf{sk}}}
\newcommand{\pk}{\ensuremath{\mathsf{pk}}}
\newcommand{\voutput}{\ensuremath{\sigma}}
\newcommand{\vproof}{\ensuremath{{\pi}}}
\newcommand{\rlen}{\ensuremath{\tau}}
\newcommand{\block}{\ensuremath{B}}
\newcommand{\cin}{\stackrel{?}{\in}}
\newcommand{\Verifiable}{\ensuremath{\textsf{Verifiable}}\xspace}
\newcommand{\PoolVerify}{\ensuremath{\textsf{PVer}}\xspace}
\newcommand{\calU}{\mathcal{U}}
\newcommand{\calS}{\mathcal{S}}
\newcommand{\calP}{\mathcal{P}}
\newcommand{\calH}{\mathcal{H}}
\newcommand{\Chain}{\ensuremath{\mathsf{C}}}
\newcommand{\state}{\ensuremath{\mathsf{st}}}
\newcommand{\bheader}{\ensuremath{\mathsf{H}}}
\newcommand{\hash}{\ensuremath{\mathsf{hash}}}
\newcommand{\Merkleroot}{\ensuremath{\mathsf{MRoot}}}
\newcommand{\MerklePath}{\ensuremath{\mathsf{MProof}}}
\newcommand{\Merkle}{\ensuremath{\mathsf{Merkle}}}
\newcommand{\cA}{\ensuremath{\mathcal{A}}\xspace}
\newcommand{\cZ}{\ensuremath{\mathcal{Z}}\xspace}
\newcommand{\ledger}{\ensuremath{\mathsf{L}}\xspace}
\newcommand{\tx}{\ensuremath{\mathsf{tx}}\xspace}
\newcommand{\Gen}{\ensuremath{\textsf{Gen}}\xspace}
\newcommand{\Verify}{\ensuremath{\textsf{Verify}}\xspace}
\newcommand{\Prove}{\ensuremath{\mathtt{Prove}}\xspace}
\newcommand{\SK}{\ensuremath{\mathsf{sk}}}
\newcommand{\PK}{\ensuremath{\mathsf{pk}}}
\newcommand{\negl}{\ensuremath{\mathsf{negl}}\xspace}
\newcommand{\secp}{\ensuremath{\kappa}\xspace}
\newcommand{\PPT}{\ensuremath{\textsc{ppt}}\xspace}
\newcommand{\rand}{\ensuremath{\mathsf{rnd}}}
\newcommand{\pool}{\ensuremath{\mathsf{S}}}
\newcommand{\sprob}{\ensuremath{c}}

%% file: intro.tex
\section{Introduction}
In recent years, Federated Learning (FL) has emerged as an auspicious large-scale distributed learning framework that simultaneously offers both high performance in training models and privacy protection for clients. FL, by design, allows millions of clients to collaboratively train a global model without the need of disclosing their private training data. In each training round, a central server distributes the current global model to a random subset of clients who will train locally and upload model updates to the server. Then, the server averages the updates into a new global model. FL has inspired many applications in various domains, including  training mobile apps \cite{hard2018federated,yang2018applied}, self-driving cars \cite{posner2021federated, khan2021dispersed}, digital health \cite{ching2018opportunities,rieke2020future}, and smart manufacturing \cite{konevcny2016federated, hao2019efficient}. 

Although training data never leaves clients' devices, data privacy can still be leaked by observing the local model updates and conducting some attacks such as membership inference \cite{shokri2017membership,salem2018ml}. Thus, FL is not particularly secure against an honest-but-curious server. To address this issue, recent research has focused on developing a privacy-preserving FL framework by devising secure aggregation on the local models \cite{bonawitz2017practical,aono2017privacy,zhang2020batchcrypt}. Specifically, it enables the server to privately combine the local models in order to update the global model without learning any information about each individual local model. As a result, the local model updates are concealed from the server, thereby preventing the server from exploiting the updates of any client to infer their private training data.%a client can share a local model update knowing the server can only see that update only after it has been averaged with those of other clients.

However, in this paper, we exploit a gap in the existing secure aggregation and show that they are inadequate to protect the data privacy. %local model updates. 
Particularly, we demonstrate that a \textit{semi-malicious} server can circumvent a secure aggregation to learn the local model updates of a victim client via our proposed \textit{biased selection attack}. Intuitively, our attack leverages the fact that the central server in FL has a freedom to select any pool of clients to participate in each training round. Hence, it can manipulate the client selection process to target the victim and extract their update from the output of the secure aggregation protocol. We present two different strategies to conduct the biased selection attack, and show experimentally that the server can successfully infer some information about the victim's private training data without making any additional security assumptions about the capabilities of the server. 

To counter this attack, we focus on strictly enforcing a random selection of clients on the central server, thereby preventing it from manipulating the selection process at its discretion. To this end, we propose using blockchain as a public trust entity and devise a verifiable random selection protocol for the server to randomly select a pool of clients in each training round. Specifically, we utilize the blockchain as a source of randomness that is used to determine the pool of clients that will participate in a training round. Via the immutability of blockchain, the clients can verify the correctness of the random selection protocol, i.e., ensuring that they are indeed randomly selected. To demonstrate the feasibility of our solution, we concretely prove that our protocol is secure against the biased selection attack. We also benchmark the performance  of the proposed protocol with   an Ehtereum-like blockchain and show that it imposes minimal overhead on FL. 

\noindent\textbf{Contributions.} Our contributions are summarized as follows:

\begin{itemize}
    \item We propose the biased selection attack where the server learns the local model updates of a victim in spite of secure aggregation. We describe two strategies to perform this attack without making extra security assumptions on the server. Then, we conduct some experiments to demonstrate its viability with respect to inferring some information about the victim's training data.
    \item As a countermeasure, we devise a verifiable random selection protocol for the server to randomly select clients in each training round. Our protocol leverages blockchain as a source of randomness so that the clients can verify whether the server correctly follows the selection protocol. Therefore, it enforces a random selection of clients, making the biased selection attack infeasible.
    \item We present concrete security proofs to show that the proposed protocol is secure against the attack. We also analyze the communication and computation cost of the protocol, together with some benchmarks to show that its overhead on FL is minimal.
\end{itemize}

\noindent\textbf{Organization.} The rest of the manuscript is structured in the following manner. Section~\ref{sec:prelim} establishes the preliminaries for our paper. We present the biased selection attack in Section~\ref{sec:attack}. Section~\ref{sec:prot} describes our proposed client selection protocol. We then provide security and performance analysis in Section~\ref{sec:analysis}. Experiments to evaluate our solution are given in Section~\ref{sec:exp}. We discuss some related work in Section~\ref{sec:rel} and finally provide concluding remarks in Section~\ref{sec:con}.

%% file: model.tex
%In this work, we consider \emph{semi malicious model}. Assume the server behaves honestly in the first step of the protocol where it commits the PKs of clients. After this point, it can deviate arbitrarily from the protocol.

%% file: federated-learning.tex
\vspace{-0.12in}
\section{Preliminaries}\label{sec:prelim}
\vspace{-0.05in}
\subsection{Federated Learning and Secure Aggregation}
\vspace{-0.05in}
% Federated Learning is a training approach for deep learning models that allows participants to collaboratively build shared models based on their private data without disclosing them. Federated learning divides deep neural network training among various participants by iteratively aggregating local models into a joint global model \cite{mcmahan2017communication}. Local training data never leaves participants' workstations, allowing federated models to train on sensitive private data, such as participants' written messages, that is significantly different from publicly accessible datasets.

Depending on how training data is distributed among the participants, there are two main versions of federated learning: horizontal and vertical. In this paper, we focus on a horizontal setting in which different data owners hold the same set of features but different sets of samples. %On the other hand, data owners in vertical federated learning possess different sets of features but the same set of samples, such that the different sets of features held by different data owners are aligned by a unique sample identifier.

Typically, an FL process follows the FedAvg framework \cite{mcmahan2017communication} which comprises multiple rounds. In this setting, a server and a set $\calU$ of $n=|\calU|$ clients participate in a collaborative learning process.
Each client $u\in \calU$ holds a training dataset $D_u$ and agrees on a single deep learning task and model architecture to train a global model. A central server $\calS$ keeps the parameters $G^t$ of the global model at round $t$. Let $x^t_u$ be a vector 
%of dimension $m$ 
representing the parameters of the local model of client $u$ at round $t$. Each training round includes the following phases:
\begin{enumerate}
	\item \emph{Client selection}:  $\calS$ samples a subset of $m$ clients $\calU' \subseteq \calU$ and sends them the current global model $G^t$.
	\item \emph{Client computation}: each selected client $u\in \calU'$ updates $G^t$ to a new local model $x^{t}_u$ by training on their private data $D_u$, and uploads $x^{t}_u$ to the central server $\calS$.
	\item \emph{Aggregation}: the central server $\calS$ averages the received local models to generate a new global model as follows:
    \begin{equation}\label{equ:globalcompute}
        G^{t+1} = \frac{1}{m} \sum_{u\in \calU'} x_u^{t}
    \end{equation}
\end{enumerate}
The training continues until the global model converges.
 
% There are two possibilities for the position of an adversary in federated learning: The adversary can be the central server, or a subset of clients. A semi-malicious server can receive updates from each individual client over time $x_u^t$, and use them to infer information about the training set of each client \cite{fredrikson2015model,abadi2016deep}. The server can also control the Client selection process, and can act actively to extract more information about the training set of a client. Alternatively, the adversary can be a subset of the clients. A malicious client can observe the global parameters over time $G^t$, and craft his own
% adversarial parameter updates $x_u^{t}$ to conduct some attacks to the global model, such as backdoor attacks \cite{bagdasaryan2020backdoor}.

To counter several attacks conducted based on the local model updates of clients, such as inference attacks by the server \cite{fredrikson2015model,abadi2016deep}, the Aggregation phase can be replaced by a \emph{secure aggregation} protocol such that each $x_u^t$ is not exposed to the server \cite{bonawitz2017practical,aono2017privacy,zhang2020batchcrypt}. By leveraging cryptographic secure multiparty computation (SMC), the secure aggregation protocols can guarantee that the server cannot learn any information about each local model update, but still be able to construct the sum of all updates. Specifically, with secure aggregation, the equation (\ref{equ:globalcompute}) is replaced by:
\begin{equation} \label{equ:secureagg}
    G^{t+1} = \frac{1}{m} \prod_{\{x_u^t | u\in \calU'\}}\left[\sum_{u\in \calU'} x_u^{t}\right]
\end{equation}
where $\prod_{X}\left[f(X,\cdot)\right]$ denotes an abstract secure computation protocol on some function $f(X,\cdot)$ and $X$ is a private input. The protocol $\prod_{X}\left[f(X,\cdot)\right]$ is: (1) \textit{correct} if it outputs the same value as $f(X,\cdot)$, and is: (2) \textit{secure} if it does not reveal $X$ during the execution of the protocol.% under some threat models.

%% file: blockchain.tex
\vspace{-0.08in}
\subsection{Blockchain}
Blockchain,  introduced in \cite{nakamoto2008bitcoin},  is a type of distributed ledger, jointly maintained by a set of nodes in a network, called miners.  Blockchain can provide guarantees on the correctness (i.e. tamper-resistance) and security of the ledger without the need of trust on a central trusted party. 
%The miners follow a consensus protocol  to maintain a block data structure.% (hence, the name blockchain).

A consensus protocol for maintaining blockchain is called secure if it satisfies the   following two security properties: 1) \emph{persistence}: all honest  miners have the same view of the ledger; and 2) \emph{liveness}: the valid transactions will eventually be added to the ledger. 

In this work, we consider a proof-of-work (PoW) blockchain, in which miners compete to solve a PoW puzzle. The miner who solves the puzzle can append a new block into a blockchain data structure. The PoW blockchain is shown to be  secure under the assumption that the honest miners hold the majority of mining power \cite{garay2015bitcoin}. The security of the protocol is parameterized by 
%the security parameter $\kappa$ 
the length of the hash function $\kappa \in \mathds{N}$ \cite{garay2015bitcoin}, called \emph{security parameter}.

The blockchain is used in our   client selection protocol to ensure 1) all clients in FL have the same views on the selected clients, and 2) a provably random selection of the client. 
\subsection{Verifiable random function}

%\noindent\emph{Verifiable random function.} 
{
To implement the provable random client selection, we use a cryptographic tool called \emph{verifiable random functions} (VRF) \cite{dodis2005verifiable}. 
VRF is  a public-key pseudorandom function that provides proofs showing that its outputs were calculated correctly and randomly, i.e., hard to predict. Consider a user with secret and public keys $\sk$ and $\pk$. The user can use VRF to generate a function output $\sigma$ and a proof $\vproof$ for any input value $x$ by running a function $\VRFprove_ \sk(x)$. 
Everyone else, using the proof $\vproof$ and the public key $\pk$, can check that the output $\sigma$ was calculated correctly by calling a function $\VRFverify(\pk, \voutput, \vproof)$. Yet, the proof $\vproof$ and the output $\sigma$ does not reveal any information on the secret key $\sk$.

In our protocol, the input value $x$ in the VRF is a randomness $\rand$, extracted from the blockchain.
Each client $i \in \calU$ independently computes an VRF output $\sigma_i$ on the input value $\rand$ to determine whether or not $i$ is selected into the pool. }

%% file: attack.tex
\section{Biased selection attack and Secure client selection problem}\label{sec:attack}

% \subsection{Threat model}
% Our threat model largely follows that of previous work in secure aggregation \cite{bonawitz2017practical,zhang2020batchcrypt}. We consider a computationally bounded adversary that can corrupt the server or any subset of clients as follows.

This section describes a simple yet effective biased selection attack and defines necessary properties of a secure client selection. First, we establish the threat model as follows.

\noindent\textbf{Threat model.} Our threat model extends that of previous work on secure aggregation in FL \cite{bonawitz2017practical,zhang2020batchcrypt,aono2017privacy}. Instead of an honest-but-curious server, we consider a \textit{semi-malicious} server that honestly follows the training protocol of FL, except that it tries to manipulate the selection process to its advantage. We assume that a secure aggregation protocol is used such that the server learns nothing other than the sum of the model updates in each training round as in equation (\ref{equ:secureagg}).
% We consider both cases where the server does or does not collude with a small subset of clients. 
The server can collude with a subset of clients. We denote by $\beta \in (0, 1)$ an upper-bound on the fraction of colluding clients.
 The goal of the server is to learn the parameters of the victim's local model updates, from which it can infer some properties about the victim's training data.

\vspace{-0.08in}
\subsection{Biased selection attacks}

% Due to the fact that the server in federated learning has full control over the clients selection in each training round, the server can potentially manipulate the selection process to violate user privacy in spite of secure aggregation. In specific, the server can learn about the parameters of a target victim's local model, essentially circumventing the privacy guarantee of secure aggregation. This section describes strategies that the server can employ to conduct such attacks.

We present two different strategies to conduct this attack. First, we show that the server can viably collude with some clients $\Bar{\calU} \subset \calU$ to learn the local models of a victim $v \in \calU \setminus \Bar{\calU}$. Second, we demonstrate that the server can still learn some information about the victim's local model $x_v$ and conduct inference attacks even without colluding with some clients. %\mt{let not call malicious clients. Call something that show these clients work with server. Malicous sound like they will attack the server}

\noindent\textbf{Colluding attack.}  Let $\Bar{\calU} \subset \calU$ be the set of clients with which the server can collude. At a particular round $t$, the server can extract the victim's local model $x_v^t$ as follows:
\begin{enumerate}
    \item The central server $\calS$ selects the victim $v$ and a subset of colluding clients $\calU' \subseteq \Bar{\calU}$ and sends them the current global model $G^t$.
    \item The selected clients compute their local model updates as normal. However, each selected colluding client $u\in \calU'$ secretly shares their $x_u^t$ with the server.
    \item The server, via a secure aggregation protocol, obtains $S = \prod_{\{x_u^t | u\in \calU' \cup v\}}\left[\sum_{u\in \calU' \cup v} x_u^{t}\right]$, which is the sum of the clients' local models as in equation (\ref{equ:secureagg}). Since the server knows the local models of the colluding clients, i.e., $x_u^t$ for $u\in \calU'$, it can extract the victim's model as $x_v^t = S - \sum_{u\in \calU'} x_u^t$. 
\end{enumerate}

Note that the server cannot solely select the victim as the only client in the training round since certain secure aggregation protocols require some form of communication between the selected clients \cite{bonawitz2017practical}. Hence, the server has to select some clients that it can collude with, i.e., $\Bar{\calU} \neq \emptyset$. In fact, those colluding clients are not necessarily real devices in the system, but can be some Sybil clients created by the server. Therefore, it is viable for the server to have some clients to collude with and conduct this attack to extract the victim's model.

\begin{figure}
    \centering
    \includegraphics[width=\linewidth]{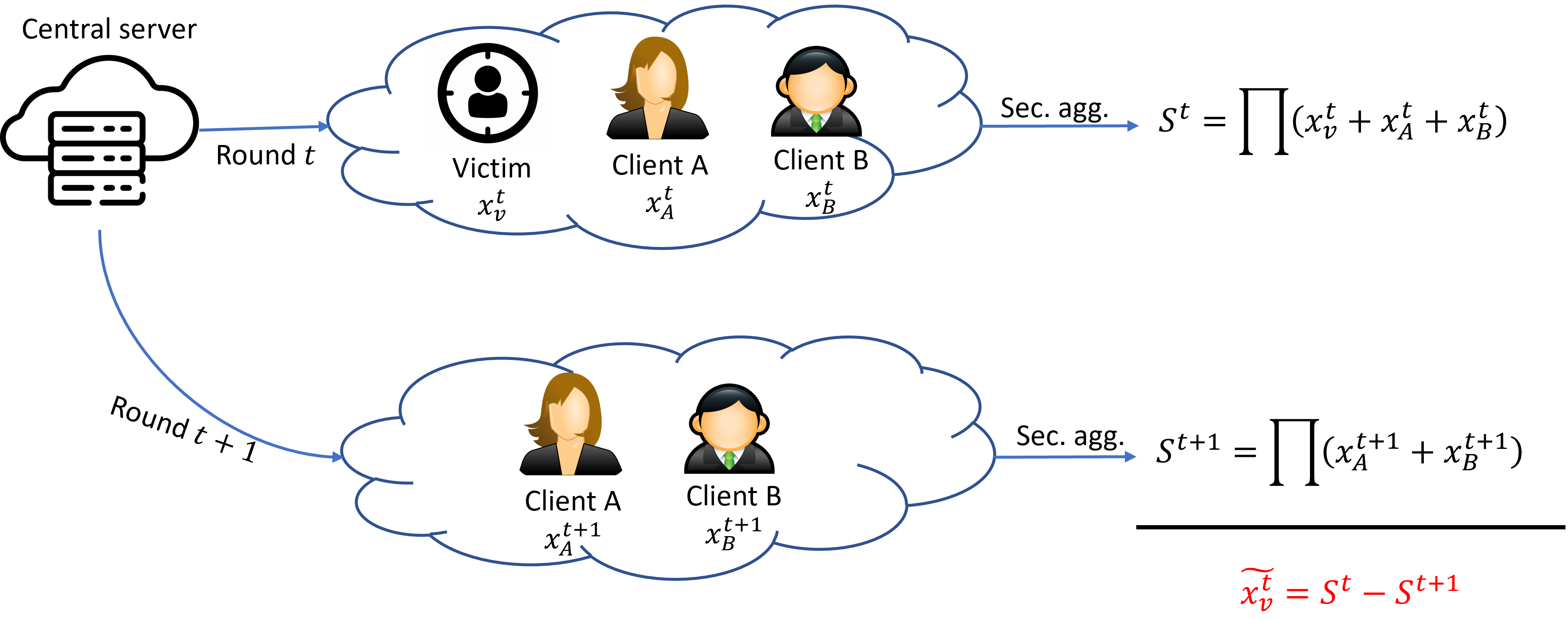}
    \caption{Overview of the non-colluding attack. The attack works in two rounds where the victim is selected in the first round, but not in the second round. The sum of the local models' parameters of each round is obtained via secure aggregation (Sec. agg.). The difference between the sums of two rounds give an approximation $\Tilde{x_v^t}$ of the victim's model.}
    \label{fig:non-collude}
    % \vspace{-0.2in}
\end{figure}

\noindent\textbf{Non-colluding attack.} In this strategy, even if we restrict the threat model to forbid collusion between the server and the clients, the server can still learn some information about the victim's model only by manipulating the client selection process. This attack requires at least two training rounds as illustrated in Fig. \ref{fig:non-collude}. The attack procedure is shown below:

\begin{enumerate}
    \item At round $t$, the server selects the victim $v$ and a subset of other clients $\calU' \subseteq \calU \setminus v$ to conduct the training round with $G^t$.
    \item Through secure aggregation, the server obtains $S^t = \prod_{\{x_u^t | u\in \calU' \cup v\}}\left[\sum_{u\in \calU' \cup v} x_u^{t}\right]$ which is the sum of the clients' models including the victim's.
    \item At round $t+1$, the server re-selects the subset $\calU'$ (without selecting $v$) and conducts the training round with $G^{t+1}$.
    \item The server obtains $S^{t+1} = \prod_{\{x_u^{t+1} | u\in \calU'\}}\left[\sum_{u\in \calU'} x_u^{t+1}\right]$ which is the sum of the clients' models excluding the victim's.
    \item The server then extracts an \emph{approximation} of the victim's model $x_v^t$ by $\Tilde{x_v^t} = S^{t+1} - S^t$
\end{enumerate}

The intuition behind this strategy is that, suppose $G^{t+1} = G^t$, for each client $u\in \calU'$, the local model parameters in the two rounds $x_u^t$ and $x_u^{t+1}$ are trained with the same initialization, same algorithm and on the same training data $D_u$. Therefore, we can expect that $x_u^{t+1} \approx x_u^t$ for $u\in \calU'$. As such, $S^{t+1} - S^t$ should give a good approximation of the victim's local model.

Regarding the assumption that $G^{t+1} = G^t$, the server can simply reuse $G^t$ at round $t+1$. Moreover, even if the server chooses not to reuse $G^t$ to avoid being suspicious, it can still honestly update $G^{t+1}$ according to the protocol and send $G^{t+1}$ to the clients at round $t+1$. However, this should be done only when the global model has already converged, thus $G^{t+1} \approx G^t$, so that the attack is still effective.

\noindent\textbf{Membership inference attack on $\Tilde{x_v^t}$.} With the non-colluding attack, we show how the server is able to obtain an approximation of the victim's model $\Tilde{x_v^t}$, the question remains whether the server can conduct any kind of privacy attacks on $\Tilde{x_v^t}$, in other words, does $\Tilde{x_v^t}$ leak any information about $D_v$? We investigate this issue by conducting some experiments with membership inference attacks \cite{salem2018ml,shokri2017membership} on $\Tilde{x_v^t}$. 

In our experiments, we use the CIFAR-10 dataset \cite{krizhevsky2009learning} and partition it among 50 clients, where each one holds about 1000 training samples. For the classification task, we use a convolutional neural network composed of two convolutional layers and two pooling layers, together with one fully connected layer, and a Softmax layer at the end. ReLU is used as the activation function. In each training round, the server randomly selects 6 clients to locally train their models using an Adam optimizer with 5 epochs and a batch size of 32. The global model converges after 100 training rounds. Then, at $t=101$, the server conducts the non-colluding attack in two rounds to obtain the approximation $\Tilde{x_v^t}$ of the victim's model.

Next, we conduct the membership inference attack from Shokri et al. \cite{shokri2017membership} on $\Tilde{x_v^t}$. This attack builds shadow models that mimic the behaviour of the victim's model (i.e., same model architecture, training data comes from the same distribution), and then uses the posteriors of those shadow models to train an attack model that determines whether a data sample is a member of the training dataset or not. In our experiments, we train the attack model as a Support Vector Machine (SVM) classifier. The attack success rate is determined by the portion of data samples that the attack model correctly predicts their membership, and it should be greater than 0.5, which is the baseline for random guessing.

\begin{figure}
    \centering
    \includegraphics[width=0.55\linewidth]{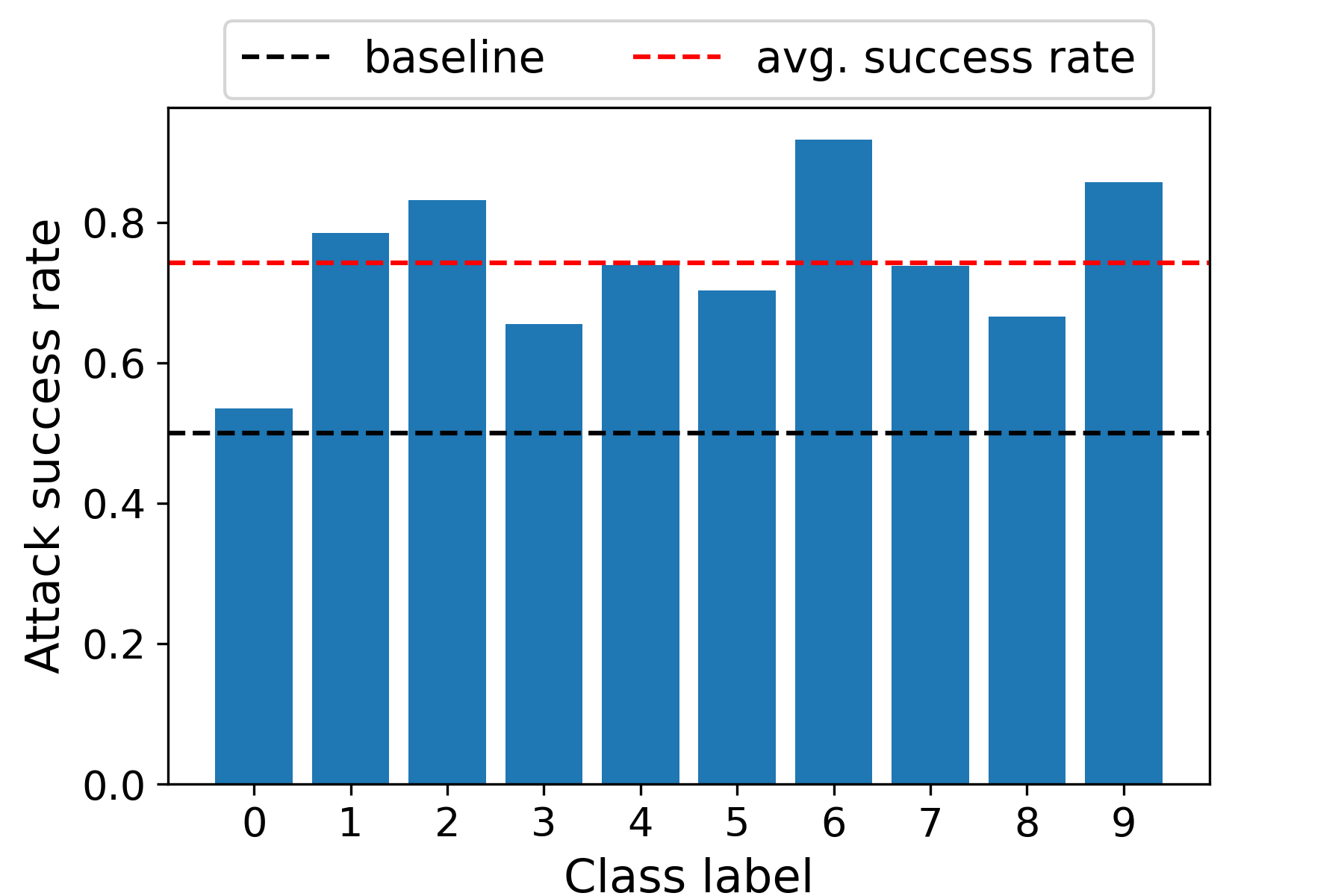}
    \caption{The success rate of the membership inference attack \cite{shokri2017membership} on $\Tilde{x_v^t}$ per class label. The black dashed line shows the baseline success rate of random guessing, which is 0.5. The red dashed line shows the average attack success rate across all class labels, which is 0.743.}
    \label{fig:membership}
    % \vspace{-0.19in}
\end{figure}

Figure \ref{fig:membership} shows the attack success rate of the membership inference attack per class label. We can see that the average success rate is about 0.743 which is much higher than the baseline. In particular, the attack can attain 0.92 success rate on the class label 6. From this result, we can conclude that even without colluding with clients, the server can still learn some information about the victim's training dataset $D_v$ only by manipulating the client selection process. %\mt{I'm done with III.A}

% \begin{figure}
%     \centering
%     \includegraphics[width=\linewidth]{figures/attack1.png}
%     \caption{Attack success on 10 classes of CIFAR10.}
%     \label{fig:my_label}
% \end{figure}

% \begin{figure}
%     \centering
%     \includegraphics[width=\linewidth]{figures/attack2.png}
%     \caption{Caption}
%     \label{fig:my_label}
% \end{figure}

% \begin{figure}
%     \centering
%     \includegraphics[width=\linewidth]{figures/attack3.png}
%     \caption{Caption}
%     \label{fig:my_label}
% \end{figure}

% \begin{figure}
%     \centering
%     \includegraphics[width=\linewidth]{figures/attack4.png}
%     \caption{Caption}
%     \label{fig:my_label}
% \end{figure}

%\begin{figure}
%    \centering
%    \includegraphics[width=\linewidth]{figures/attack1.png}
%    \caption{Attack success on 10 classes of CIFAR10.}
%    \label{fig:my_label}
%\end{figure}
%
%\begin{figure}
%    \centering
%    \includegraphics[width=\linewidth]{figures/attack2.png}
%    \caption{Caption}
%    \label{fig:my_label}
%\end{figure}
%
%\begin{figure}
%    \centering
%    \includegraphics[width=\linewidth]{figures/attack3.png}
%    \caption{Caption}
%    \label{fig:my_label}
%\end{figure}
%
%\begin{figure}
%    \centering
%    \includegraphics[width=\linewidth]{figures/attack4.png}
%    \caption{Caption}
%    \label{fig:my_label}
%\end{figure}

%% file: problem.tex
% !TEX root = main.tex

%\vspace{-0.08in}
\subsection{Secure client selection problem}
We define a new problem, called \emph{secure client selection (SCC)} problem that asks for a protocol $\Pi$, executed by a server $\calS$ and a set of clients $\calU$, to select subset of clients  in each training round in FL. At the end of the execution of the protocol $\Pi$, for each client $j$, the server $\calS$ sends  a collections of proofs $\{\omega_j^{(i)} \}_{i\in \calU}$, in which $\omega_j^{(i)}$  is either empty or \emph{a proof on whether or not the client $i$ is selected}.
The designed protocol is required to have three  security properties, namely, \emph{pool consistency, pool quality, and anti-targeting}.

%We say $\Pi$ is secure if the following properties hold: 1) \emph{pool consistency}: the server cannot prove that a client is both selected and not selected to honest clients; 2) \emph{pool quality}: the fraction of honest selected clients is bigger than a given parameter; and 3) \emph{anti-targeting}: all honest clients have the same probability to be selected. 

%Formally, we define the secure client selection problem as follows.
%In each training round, the server obtains a local state $\state_\calS$, and each client $i$ obtain a local state $\state_i$. 
%\pnote{The server send a proof $\omega_j^{(i)}$ to a client $j$.}
\begin{definition}[Secure client selection problem]
	\label{def:sec}
	 Let $\state_j$ be the local state of the client $j \in \calU$  at the end of a training round.	We say  $\Pi$ is secure if{f} there exists a predicate $\PoolVerify_\Pi$ that takes the state $\state_j$ of a client $j$,
% 	(the verifier),
% 	a client $i$,
	a proof $\omega_j^{(i)}$ (provided by server $\calS$)  as input and output 
%	 the state $\state_j$, a client $i$, an proof $\omega_i$ as input and  i.e., 
	\begin{equation*}
	\PoolVerify_\Pi(\state_j,\omega_j^{(i)}) = 
	\begin{cases}
	1 &\text{ if } i \text{ is selected in the view of } j ,\\
	0 &\text{ if } i \text{ is not selected in the view of } j,\\
	\perp &\text{ if } \omega_j^{(i)} \text{ is empty or invalid.}
	\end{cases}
	\end{equation*}
	
%	\mt{I'm confused with $\omega_j^{(i)}$}
%	 a subset of clients $\calP \subset \calU$ is selected to participate in the training. 
%	After the clients are selected, for any client $i \in \calU$, the server can provide a prove $\omega_i$ to verify if a client $i$ is selected, i.e., 
%	\begin{equation*}
%		\PoolVerify_\Pi(\state_j,i,\omega_i) = 
%		\begin{cases}
%			1 &\text{ if } i \text{ is selected} \\
%			0 &\text{ if } i \text{ is not select} \\
%			\perp &\text{ if } \omega_i \text{ is invalid}
%		\end{cases}
%	\end{equation*}
% 	After the clients are selected, the server provide a valid proof $\omega_i$ to each client $i$, i.e., $\PoolVerify_\Pi(\state_i,\omega_i) \ne \perp$.
%	we can verify if a client $u$ is selected using a proof $\omega_u$
	with the following properties: 
	\begin{itemize}
		\item \emph{Pool consistency}: $\forall i \in \calU$ and  $\forall j_1, j_2 \in \calH$,
		$$\Pr\left[		
		\begin{array}{l|l}
	\multirow{2}{*}{$\exists \ \omega_{j_1}^{(i)},\omega_{j_2}^{(i)}$} &
	 \PoolVerify_\Pi(\state_{j},\omega_{j_1}^{(i)}) = 1\land\\
		& \PoolVerify_\Pi(\state_{j'},\omega_{j_2}^{(i)}) = 0
		\end{array}
		\right] \le e^{-\Omega(\kappa)},$$
%		Let $\Verifiable_v(u \in \calP)$ be the event where the server $\calS$ can  prove to the client $v$ that $u \in \calP$ and $\Verifiable_v(u \notin \calP)$ be the event where the server $\calS$ can  prove to the client $v$ that $u \notin \calP$.
%		For any client $u$ and any honest clients $v_1,v_2$, we have
%		$$ \Pr[\PoolVerify_\Pi(u,\omega_u) = 1\land \PoolVerify_\Pi(u,\omega'_u) = 0] \le \negl(\kappa),$$
		where $\calH$ denotes the set of honest clients and $\kappa$ is the security parameter.
%		\mt{what are domain and range of $\kappa$. Look like it's just an integer.. How to set this, did we evaluate the performance based on $\kappa$?}
%		 and $\negl(.)$ is a negligible function. 
%		In each training round, all honest clients have the same view on the pool. 
		\item \emph{$\gamma$-pool quality for $\gamma \in (0, 1)$}:
%		 (parametrized by $\gamma$):
		Let $\calP$ be the set of selected clients, defined as:
		$$ \calP = \{i \in U: \exists  j \in \calH \text{ s.t. } \PoolVerify_\Pi(\state_j,\omega_j^{(i)}) = 1\}.$$
%		 The fraction of honest clients in the pool is at least $\gamma$. 
		 We have:
		$$\Pr\left[\frac{\calH \cap \calP}{\calP}  \ge \gamma\right] \ge 1-  e^{-\Omega(\kappa)}.$$

		\item \emph{Anti-targeting}:
		Let  $c= \frac{m}{n}$, termed the \emph{selection probability}.  We have:
		$$ |\Pr[i \in \calP] - c| \le e^{-\Omega(\kappa)},\ \forall i \in \calU.$$
% 		(parametrized by $\sprob$): For any honest client $i \in \calU$, we have,
% 		$$\Pr[ \exists \text{ honest } j,\omega_j^{(i)} \text{ s.t. } \PoolVerify_\Pi(\state_j,\omega_j^{(i)}) = 1] = \sprob.$$
	\end{itemize}
\end{definition}

%\mt{pool quality depends on $\gamma$. What is this $\gamma$, how to define etc...?}
The pool consistency ensures that the server cannot prove that a client is  selected to one client while proving that it is not selected to another  client. The pool quality enforces a minimum fraction of selected honest  clients. Finally, the anti-targeting guarantees that all honest clients are selected with the almost the same probability. 

%Consider a set of $n$ clients. I a pool selection protocol, 
%We require the following properties from a secure pool selection protocol:
%\begin{itemize}
%	\item \emph{Pool consistency}: In each training round, all honest clients have the same view on the pool.  
%	\item \emph{Pool quality}: The fraction of honest clients in the pool is at least $\rho$. 
%	\item \emph{Anti-targeting}: The probability of selecting any honest client into the pool is the same and equals $\sprob$, we $\sprob \in (0,1)$ is the fraction of clients the server want to select in each round. 
%\end{itemize}
%
%Here, the pool consistency and pool quality properties are required to guarantee the security of the aggregation. The anti-targeting property ensure the membership attacks. 

\noindent \textbf{Baseline protocol.} Initially, the clients register their public keys on the blockchain. 
%After the registration, the server and the clients can obtain a randomness by taking the hash value of the concatenation of all public keys. 
In each training round, each client computes a set of selected clients using a pre-arranged function of the registered information and the round number. The function can be implemented using blockchain's smart contracts so that all the clients agree on the same list of selected clients.
%each client can locally compute the 

The above protocol provides pool consistency. However, it could not guarantee either pool quality or anti-targeting properties.
Jumping a head, in Section \ref{sec:analysis}, we will show that our protocol can achieve all three security properties of SCC problem (see Table \ref{tab:summary}).
% as shown in the following paragraph.
\begin{table}[htp!]
	\centering
	%	\scriptsize
	\begin{tabular}{| c| ccc|}
		\hline
		%		\multirow{ 2}{*}{Protocols}&\multicolumn{3}{|c|}{Security\\
		%		\cline{2-6}
		Protocols & PC & PQ &AT\\
		\hline
		Baseline & yes & no & no  \\ 
		This work & yes&yes&yes \\
		\hline
	\end{tabular}
	\caption{The security of the baseline protocol and our protocol. PC, PQ, and AT stand for pool consistency, pool quality, and anti-targeting, respectively. 
	}
	\label{tab:summary}
\end{table}

\emph{Grinding attack on the baseline protocol.} 
%As we mentioned above, when the randomness is obtained, we can predict which clients are selected at any round. 
The clients, who colluded with the server, can wait til all other clients complete their registration. Then they can probe for different public keys to bias the selection of clients. Since the round number and the registration information of honest clients are known, the search can be done to either give more chance for colluding clients to be selected (breaking the pool quality) or more chance towards a targeted honest client (breaking the anti-targeting).

%We start with a simple baseline protocol. Initially, clients register their public keys on the blockchain. Then, in each training round, the server uses a current round number to select a subset of clients on the blockchain. 
%In the baseline protocol, the server can predict which clients are selected at any round. Further, since the information of all selected clients are stored on the blockchain, it requires a large amount of storage.   

\noindent \textbf{Completeness of the security properties.} We show that if a protocol satisfies the above three security properties in Definition \ref{def:sec}, the adversary cannot perform biased selection attacks as discussed in the previous subsection.

\begin{lemma}
	Consider a pool selection protocol that can achieve pool consistency, pool quality, and anti-targeting properties. The probability that the server can perform colluding/non-colluding attacks is at most $e^{-\Omega( \min\{h, \kappa\})},$ where $h = |\calH|$ and $\kappa$ is the security parameter. 
% 	\begin{itemize}
% 		\item The probability that the server can perform colluding attack is at most $e^{-\Omega(h)} + e^{-\Omega(\kappa)}$
% % 		$c (1-c)^{h-1} + e^{-\Omega(\kappa)}  = e^{-\Omega(h - \log(h))} + e^{-\Omega(\kappa)} $.
% 		\item The probability that an server can perform a non-colluding attack is at most 
% % 		$h \cdot c^{s-1} (1-c)^{h-s-1} + e^{-\Omega(\kappa)} =
% 		$e^{-\Omega(h)} + e^{-\Omega(\kappa)}$
% 	\end{itemize}
\end{lemma}
Due to the space limit, we provide an outline of our proof. By pool consistency, all nodes have the same view on $\calP$, the set of selected clients  with a probability at least $1-e^{-\Omega(\kappa)}$.

From the anti-targeting property, the probability that an honest client is selected concentrate around $c$, the selection probability.
%Let $c$ be the probability that an honest client is selected (see the anti-targeting property). 
We have
% The probability that exact one honest client is selected is
% 	$$h\cdot c (1-c)^{h-1} = e^{-\Omega(h)}$$ 
\begin{align*}
    \Pr[\text{Colluding attack}] &\leq \Pr[\calH \cap \calP =1] \\
    &\approx {h \choose 1} c (1-c)^{h-1} = e^{-\Omega(h)}.
\end{align*}
% 	Based on the pool quality property, the fraction of honest selected clients is smaller than $\gamma$ with probability at most $\negl(\kappa)$. Thus, the server can perform a colluding attack with probability at most $\negl(\kappa)$.
% 	With the pool consistency, the probability that the server can perform a colluding attack is
For non-colluding attacks, let $S_i, S_{i+1}$ be the sets of selected  clients in two consecutive rounds and $s = |S_i|$.   \begin{align*}
    \Pr[\text{Non-colluding attack}] &\leq \Pr[S_{i+1} = S_i \cup \{x\}, x\in \calH]\\ &\approx {s \choose 1}  c^{s-1} (1-c)^{h-s-1}  = e^{-\Omega(h)}.
\end{align*}
Combining all the probabilities of the bad events yields the bound $e^{-\Omega( \min\{h, \kappa\})}$.
%\vspace{-0.1in}

%% file: client-selection.tex
\section{Client selection protocol}\label{sec:prot}

%\subsection{Overview}
%\subsection{Based line protocol}

%\pnote{Add a baseline protocol here.}

%\subsection{Our secure pool selection protocol}

\begin{figure}[htp!]
	\centering
	\includegraphics[width=\linewidth]{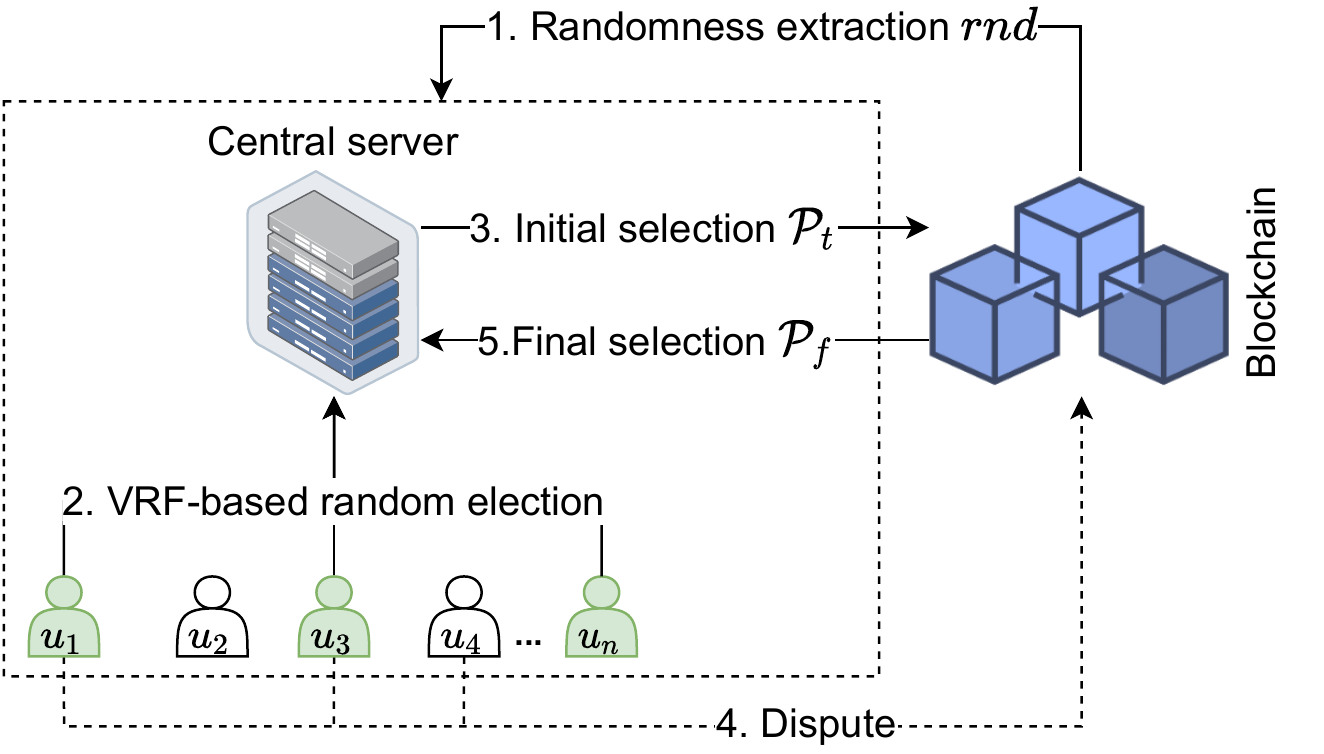}
	\caption{The $5$ steps of client selection in each training round: (1) Randomness extraction, (2) Random election, (3) Initial commitment, (4) Dispute, and (5) Final selection.}\label{fig:pool-selection}
% 	\vspace{-0.15in}
\end{figure}

Our protocol consists of a one-time registration phase and multiple training rounds. In the registration phase, 
 clients registered their public keys with the server.
At the beginning of each training round,  a random subset of  clients will be selected using our client selection protocol.

\textsf{Registration phase}. At the beginning of the FL process, each client registers its public key with the  server. The server composes a list of all public keys, submits a succinct commitment of the list to the blockchain, and provides each client with a membership proof using  Merkle proof \cite{dahlberg2016efficient}, a lightweight way to prove the membership in large sets.
% Upon gathering all public keys from the clients, 
% the server constructs a Merkle tree $\Merkle_r$ from those public keys and submits a registration transactions that consist of the Merkle tree root $\Merkleroot_r$. 
% server gathers the public keys of the clients, constructs a Merkle tree, and submits a registration transactions that consist of the Merkle tree root. Later, in each training round, a subset of clients will be selected from the set of registered clients. 
\\
\emph{Membership proof with sorted Merkle tree}. Given a set of $l$ values $X=\{d_1, d_2, \ldots, d_l\}$, a Merkle tree is a binary tree constructed over the hash values of $d_i$. The root of the Merkle tree, denoted by $\Merkleroot(X)$, can be used as a succinct representation of all the values. Knowing $\Merkleroot(X)$,  we can construct a membership/non-membership proof of size $O(\log l)$  to prove whether a value $x$ appears in $X$. Such a proof, denoted by $\MerklePath(x \cin X)$, can also be verified in a time $O(\log l)$.

Let $\calU$ be the set of registered clients, the server submits a registration transaction containing the Merkle root $\Merkleroot(\calU)$ to blockchain and sends $\MerklePath(\pk_i \cin \calU)$ to each client $i$.

\textsf{Client selection phase.} 
As shown in  Fig.~\ref{fig:pool-selection}, the selection consists of: randomness extraction, VRF-based random election, initial selection, dispute, and final selection. In randomness extraction, all clients and the server compute locally a random token $\rand$ by hashing together the block headers in the previous round. Verifiable random functions (VRFs) \cite{dodis2005verifiable}, taking the client's public key and $\rand$ as inputs, are employed to determine which clients are selected. In the initial selection, the server composes a list of selected clients and commits the list to the blockchain. A selected client can submit a dispute transaction if s/he is not properly included in the initial list, forcing the server to include her/him in the final selection. 
%\mt{how does she know to dispute, sit there and verify the initial list each round? What is the property of dispute?}

% After that, the server can provide a proof of size $O(\log k)$ to show whether or not a public key of a client is included in the Merkle tree.

% In our protocol, each training round  lasts for $\rlen = \Omega(\kappa)$ blocks on the blockchain. 
% We denote $t_r$ as the starting block of the round $r$ (we have $t_r = t_{r-1} + \rlen$). 

%Initially, clients register their public keys on the blockchain.
%Then, in each training round, the clients use the registered key to compute the VRFs on a public randomness. If the outputs of the VRFs are smaller than a threshold, the clients are selected write their public keys on the blockchain.   

%Note that, if the number of clients is big, writing their public keys on the blockchain is expensive. 
%To reduce the blockchain chain cost, we use Merkle tree to compress the public keys to a small commitment (i.e., the Merkle tree root) and only write the commitment on the blockchain. 

%We remark that, in our protocol, the server needs to maintain a full blockchain node, while each client only needs to maintain a lightweight SPV node.

We provide the details for the  steps in the client selection protocol in Alg.~\ref{alg:protocol}.
We use the height of the blockchain, or \emph{block height} (BH), to measure time. 
We select a parameter 
$\tau = \Omega(\kappa)$ so that  sent messages are received and submitted transactions are finalized within $\tau$ blocks. For a training round started at BH $\ell$, the client selection protocol is executed  between block heights $\ell $ and $\ell + 2 \tau$. 

\begin{algorithm}
	\small 
	\SetKwInOut{States}{Server's state}
	\SetKwInOut{Statec}{Client $i$'s state}
	\SetKwInOut{State}{One trainning round}
	\SetKwInOut{Goal}{Goal}
	\SetKwInOut{Input}{Input}
	\SetKwInOut{Output}{Output}	
	\SetKwFor{Event}{}{}{}	
	\State{Consider a training round that starts at block height (BH) $\ell$.}
	\underline{BH $\ell$}:  \emph{Randomness extraction} The server and clients extract the randomness $\rand$ from the blockchain.  \\
	\underline{BH $\ell$}: \emph{VRF-based random election}: Each client $i$ computes $(\voutput_i,\vproof_i) \gets \VRFprove_{\sk_i}(\rand)$. If $\voutput_i < c2^{\kappa}$, the client $i$ is qualified and sends the proof  $(\voutput_i,\vproof_i,\pk_i)$ to the server. \\
	%[BH $\ell + \tau$] 
	\underline{BH $\ell+\tau$}: \emph{Initial  selection}: The server submits the Merkle tree root on the set of qualified clients $\calP_t$ and sends the Merkle proof to each client.\\
% 	Let $\calP_t$ be the set of public keys of qualified clients that are verified by the server. The server submits an initial selection transaction that consists of the Merkle tree root $\Merkleroot(\calP_t)$,
% 	Upon receiving a valid proof  $(\voutput_i,\vproof_i,\pk_i)$, the server adds $\pk_i$ to $\calP$. Then, the server constructs a Merkle tree on $\calP$, submit a transaction the consists of the Merkle root, 
% and sends the Merkle proof to each clients.\\
	%[BH $\ell+2\tau$] 
	 \underline{BH $\ell+\tau$}: \emph{Dispute}: 
%	After $k = \omega(\kappa)$ blocks,
	 If a qualified client $i$ does not receive the proof from the server,
%	is not included in $\calP_t$, 
	it submits a dispute transaction that consists of the proof $(\voutput_i,\vproof_i,\pk_i)$ to blockchain.\\
	\underline{BH $\ell+2\tau$}: \emph{Final selection}: The server submits the Merkle tree root on the set of dispute clients $\calP_f$.
% 	The server submits a transaction to obtain the list of dispute clients.
	\caption{Client selection protocol.}
	\label{alg:protocol}
\end{algorithm} 
	\noindent \emph{1. Randomness extraction.} We follow the scheme to extract the randomness in \cite{david2018ouroboros}. 
% 	At the round $r$, 
	At block height $\ell$, the server and all clients compute a randomness $\rand$ by hashing together the block headers of $\kappa$ blocks created during the previous training round. The chain quality of the blockchain means that, with high probability, at least one of those blocks must be from an honest miner \cite{garay2015bitcoin}. Thus, $\rand$ includes at least one unbiased random source.
% 	More concretely, the randomness $\rand$ is computed as the hash value of the first $\kappa$ blocks in the previous training round, i.e.,  
%$$\rand \gets \hash(\bheader(\block_{t_{r-1}}),\bheader(\block_{t_{r-1}+1}),\cdots,\bheader(\block_{t_{r}-\kappa})).$$
% $$\rand = \hash(\block_{t_{r-1}},\block_{t_{r-1}+1},\cdots,\block_{t_{r-1}+\kappa}),$$
% where $t_{r-1}$ is the first block in the training round $r-1$. 
%Since the randomness is extracted only using the block headers, the clients with SPV nodes can also extract the randomness.

		\noindent \emph{2. VRF-based  random election.} After extracting the randomness, each client $i$ uses the VRF to check whether or not she/he is selected in this round. 
%		 A client $i$
%		with a secret key $\sk_i$ can call $\VRFprove_ {\sk_i}(\cdot)$ to generates a pseudorandom output $\voutput_i$ along with a proof $\vproof$. 
% 		In detail, 
		The client $i$  computes 
		the output $\voutput_i$ and the proof $\vproof_i$ of the VRF based on the randomness $\rand$, i.e., 
		$(\voutput_i,\vproof_i) \gets \VRFprove_{\sk_i}(\rand)$.
		If the VRF output $\voutput_i$ is smaller than a given threshold, i.e., $\voutput_i < c 2^{\kappa}$, the client $i$ is \emph{qualified} to be selected.
%		 If the output $\voutput_i$ is smaller than a threshold $c$, the client $i$ is selected.
		  Here,  $c=\frac{m}{n}$ is the \emph{selection probability}, i.e., the fraction of selected clients per round. 
%		  chosen based on the fraction of clients that we select in each round. 
		  If the client $i$ is qualified, she/he sends a message $(\voutput_i,\vproof_i,\pk_i)$ to the server. 
%		  \mt{what motivates $i$ to check?}
		  
%		  Other parties that have the proof  and the corresponding public key $\pk_i$ can check that the output has been generated by VRF, by calling $\VRFverify_{\pk_i}(\cdot, \voutput_i, \vproof_i)$. 
		
		\noindent \emph{3. Initial  selection}. Let $\calP_t$ be the set of public keys of qualified clients that are verified by the server. The server submits an initial selection transaction that consists of the Merkle tree root $\Merkleroot(\calP_t)$ to the blockchain.
% 		The server collects a list $\calP_t$ that consists of all public keys of qualified clients, computes a small commitment using the sorted Merkle tree, and submit an initial selection transaction to the blockchain. More precisely, upon receiving a proof $(\voutput_i,\vproof_i,\pk_i)$ from the client $i$, the servery verifies 1) $\VRFverify_{\pk_i}(\rand, \voutput_i, \vproof_i) = 1$ and 2) $\voutput_i < c 2^{\kappa}$ before adding $\pk_i$ to the list $\calP_t$. 
% 		Then, the server constructs a sorted Merkle tree $\Merkle_t$ based on the list  $\calP_t$
% 		sorts the list $\calP_t$ and computes a sorted Merkle tree $\Merkle_t$ where the leaves is indexed based on the public keys of the clients in $\calP_t$. 
%		Based on the list $\calP$, the server computes a sorted Merkle tree $\Merkle$ where the leaves is indexed based on the public keys of the clients in $\calU$. 
%		by assigning the values of the corresponding leaves of the clients in $\calP$ as $1$. 
%		as follows.
%		For each client $i \in \calP$, we add $(\pk_i,1)$ as the leaf of the Merkle tree. For each client $i \in \calU \setminus \calP$, we add $(\pk_i,0)$ as the leaf of the Merkle tree.
%		 where the leaves is indexed based on the public keys of the clients in $\calU$. 
%		each leaf 
%		in the Merkle tree is a public key in $S$. 
% 		and  submits a initial selection transaction that consists of the Merkle tree root $\Merkleroot_t$ to the blockchain. 
		After the transaction is included to the blockchain, the server provides a Merkle proof $\MerklePath(\pk_i \cin \calP_t)$  for each client $i \in \calU$.
% 		to show whether or not $\pk_i$ is included in the Merkle tree.
%		 a proof that consists of two parts: 
% 		 the Merkle tree path to prove whether or not the client $i$ is included in the commitment.
%		1) the initial  commitment transaction and the proof (to SPV node) that the transaction is included on the blockchain; 2) the Merkle tree path to prove whether or not the client $i$ is included in the commitment.
% 		If $i$ is included in the commitment, this proof is path $\MerklePath_t(\pk_i)$ from the root to $pk_i$ in $\Merkle$ to all client $i \in \calP_t$. 
% 		If $i$ is not included in the commitment this proof consists of two paths 
% 		$\MerklePath_t(\pk_{i_1})$, $\MerklePath_t(\pk_{i_2})$, where $\pk_{i_1} < \pk_i < \pk_{i_2}$ and $\pk_{i_1}$, $\pk_{i_2}$ are two consecutive leaves in the Merkle tree.
%		of whether or not the client is included in the commitment.
%		The proof consists of two parts: 1) the initial  commitment transaction and the proof that the transaction is included on the blockchain, 2) the path $\MerklePath(\pk_i)$ from the root to $pk_i$ in $\Merkle$ to all client $i \in \calP$
%		that consists of the transaction is included to the blockchain
%		the server consists of a proof that the transaction is included to the blockchain and the path $\MerklePath(\pk_i)$ from the root to $pk_i$ in $\Merkle$ to all client $i \in \calP$. 
		
		\noindent \emph{4. Dispute.} If a qualified client $i \in \calU$ does not receive any Merkle proof from the server, or finds any discrepancy between the Merkle root obtained from the server to the one that the server submitted to the blockchain, it will start a dispute process. The client will submit proof of qualification directly to blockchain to force the inclusion of itself into the pool.
		More concretely,  at block height $\ell+ \tau$, 
%		after sending the proof to the server for $k = \Omega(\kappa)$ blocks,
% 		, the qualified client $i$ waits for $\kappa$ blocks on the blockchain to ensure the that initial selection transaction is included in the blockchain.
% 		During this time, upon receiving a proof from the server, the client verifies that 1) the initial commitment transaction is included on the blockchain (using the header chain), and 2) the Merkle tree path is correct.
%		 Note that, after $\kappa$ blocks on the blockchain, we can guarantee that the any valid transaction from the server is included on the blockchain. 
% After $\kappa$ blocks, 
the client can submit a transaction containing the tuple
% the VRF output and the public key 
$(\voutput_i,\vproof_i,\pk_i)$ to the blockchain. 
The client $i$ also includes the Merkle proof $\MerklePath(\pk_i \cin \calU)$ to show that its public key is  registered. 
%\mt{I don't see the time window for $i$ to dispute (in this discussion nor in the Alg 1 pseudo). And I don't see how $i$ will do this? What motivate $i$ to sit there, to verify, to dispute? Dispute can be done for something in the past (like dispute your credit card transactions. But this is live disupte... Unless there is a clear motivation, why client bothers to do this step?}
% Note that, only the registered clients can submit the transactions. Thus, the clients $i$ also needs to include a Merkle proof to show that $\pk_i$ is included in the registration transaction. 

%If the server cannot provide the proof to show that the public key of the client $i$ is committed on the blockchain, the client $i$ can submit a dispute transaction that consists of $(\voutput_i,\vproof_i,\pk_i)$ to the blockchain. If the transaction is included to the blockchain before the client selention phase ends, the client $i$ is also considered to be selected in the current round.
	
\noindent \emph{5. Final selection.} 
At block height $\ell+ 2\tau$, the server submits a final selection transaction that contains the information of \emph{all dispute transactions}.
Let $\calP_t$ be the set of the public keys of \emph{dispute clients}, i.e., the clients who submitted dispute transactions.
% \tn{should we modify this notation?}
% Finally, the server submits a final selection transaction to include all of the \emph{dispute clients}, i.e., the clients who submitted dispute transactions. 
% More precisely, the server collects a list $\calP_f$ that consists of all public keys of the dispute clients. (It is possible that a clients belong to both $\calP_t$ and $\calP_f$.)
Then, similar to the initial selection, the server constructs a Merkle tree $\Merkle_f$ based on $\calP_f$.
% the server sorts the list $\calP_f$ and computes a sorted Merkle tree $\Merkle_f$ where the leaves is indexed based on the public keys of the clients in $\calP_f$. 
%After the dispute resolution steps, the server collects a list $\calP_f$ that consists of all public keys of the clients who submitted dispute transactions. (Note that, it is possible that a clients belong to both $\calP$ and $\calP_f$.) Based on the list $\calP_f$, the server computes a sorted Merkle tree $\Merkle_f$ the leaves is indexed based on the public keys of the clients in $\calP_f$. 
%		 by assigning the values of the corresponding leaves of the clients in $\calP$ as $1$. 
%		where each leaf in the Merkle tree is a public key in $\calP_f$. 
% Next, 
The server submits a final selection transaction that consists of the Merkle tree root $\Merkleroot(\calP_f)$ and sends a Merkle proof $\MerklePath(\pk_i \cin \calP_f)$  to each client $i \in \calU$.
% root of the Merkle tree $\Merkleroot_f$ to the blockchain. 
Here, before adding the final selection transaction to the blockchain, the miners  verify that all public keys of the dispute clients are included in the  $\Merkleroot_f$. 
% Recall that, as the information of the dispute clients can be obtained from the dispute transactions, the miners can also compute the Merkle tree $\Merkle_f$. 
% We remark that, the server must submit the 
 The \emph{correctness will be enforced} through smart contracts, executed by all miners in the blockchain.
%\end{itemize}
%1. First, the clients and the server extract a randomness from the blockchain (by concatenating the hash values of $\kappa$ consecutive blocks in the previous round.) 

%At the beginning, the server commits a transaction to start the round. Now, based on the block $B$ that contains the transactions, everyone can extract a randomness from the blockchain (by concatenating the hash values of $\kappa$ consecutive blocks that end at $B$. 

%2. Each client takes the randomness as the input of VRF to check whether or not she/he is selected. If a client is selected, she/he will send the proof to the server.
%
%3. The server then commits the list of selected clients to the blockchain (using Merkle tree) and sends each selected client proof to show that she/he is selected.
%
%4. If a client is selected but ignored by the server, she/he can commit a transaction to the blockchain. After the transaction is committed, the client can participate in that round.

%% file: analysis.tex
% \vspace{-0.1in}
\section{Security Analysis} \label{sec:analysis}

In this section, we analyze the security of our protocol in Algo. \ref{alg:protocol}. We start with the construction of our $\PoolVerify(\cdot)$ function, followed by the proof sketches on the  three security properties, defined in Section~\ref{sec:prelim}.
\emph{Pool membership verification function.} We describe the function $\PoolVerify(\state_j,\omega_j^{(i)})$ that verifies if the client $i$ is selected in the view of the client $j$. 
%\mt{inconsistent notation, earlier is $\omega_j^{(i)}$. The whole section has this problem.}

% \emph{Proof of inclusion/non-inclusion.} We describe the proof of inclusion/non-inclusion. 

 For each client $j$ with the state $\state_j$, 
%  can use the proof $\omega_i$ to verify if the client $i$ is selected. More precisely, 
%For each client $j$ with the state $\state_j$, 
the function $\PoolVerify(\state_j,\omega_j^{(i)})$ extracts the blockchain $\Chain_j$ from the local state $\state_j$ and then proceeds as follows.
\begin{itemize}
% 	\item . 
	\item The function verifies whether or not 
	(1) $\VRFverify(\pk_i, \voutput_i,\vproof_i) = 1$.
	(2) the initial selection transaction and the final commitment transaction are included in the header blockchain $\Chain_j$.
%	if 1) the format of $\omega_i$ is correct; 2) $\Verify(\voutput_i,\vproof_i,\pk_i) = 1$; 3) $\Merkleroot$ and $\Merkleroot_f$ are included in the transactions from the server.
	 If those conditions do not hold, it returns  $\perp$.
%	If the format of $\omega_i$ is not correct or $\Verify(\voutput_i,\vproof_i,\pk_i) = 0$, the function returns $\perp$.
	\item If all conditions hold, i.e., the proof $\omega_j^{(i)}$ is valid, 
%	the function verifies if $\Merkleroot$ and $\Merkleroot_f$ are stored on the blockchain.
%	Then, 
	the function verifies 
	(1) $\voutput_i < c2^\kappa$, and
	(2) $\pk_i$ is included in $\Merkleroot$ or in $\Merkleroot_f$.
	If those conditions hold, it returns $1$, i.e., the client $i$ is selected.
	
%	uses $\MerklePath(\pk_i)$ and $\MerklePath_f(\pk_i)$ to verify if $\pk_i \in \calP$ and $\pk_i \in \calP_f$.
%	If $\voutput_i < c2^\kappa$ and $\pk_i \in \calP \cup \calP_f$, the function returns $1$.
	\item Otherwise, the function returns $0$, i.e., the client $i$ is not selected. 
\end{itemize}

Recall that in our protocol, for each qualified client $j$, the server provides only   $\omega_j^{(j)}$, the  proof of membership of $j$. 
%For each client $i$, 
% After the finishing the client selection in each training round, the server can provide 
The proof  consists of 
(1) the VRF output and the public key of $j$ $(\voutput_j,\vproof_j,\pk_j)$,
(2) the initial selection transaction that consists of $\Merkleroot_t$,
%and the proof (to SPV node) that the transaction is included on the blockchain,
(3) the Merkle proof $\MerklePath_t(\pk_j)$,
% tree path to prove whether or not the client $i$ is included in the commitment $\Merkleroot_t$,
(4) the final selection transaction that consists of $\Merkleroot_f$,
% and the proof (to SPV node) that the transaction is included on the blockchain,
and (5)  the Merkle proof $\MerklePath_f(\pk_j)$.
% the Merkle tree path to prove whether or not the client $i$ is included in the commitment $\Merkleroot_f$,
% = \langle (\voutput_i,\vproof_i,\pk_i), \Merkleroot, \MerklePath(\pk_i), \Merkleroot_f,$ $\MerklePath_f(\pk_i) \rangle$ for the client $i$.

\emph{Pool from all clients' views.} We say a client $i$ is selected if there exists an honest client $j$ and a proof $\omega_j^{(i)}$ such that 	$\PoolVerify(\state_{j},\omega_j^{(i)}) = 1$. Let $\calP$ be the set of selected clients, i.e.,
$$ \calP = \{i: \exists \text{ honest client } j, \omega_j^{(i)}, s.t., \PoolVerify(\state_{j_1},\omega_j^{(i)}) = 1\}.$$
%First, we consider the client selection protocol in which the randomness is perfect. 

%We are now ready to prove the security of our pool selection protocol.
We first prove that all honest clients have the same view on the set $\calP$ of selected clients. Intuitive, as the blockchain maintains an immutable ledger, all honest clients have the same view on the commitment transactions. Thus, they can extract the same list of selected clients. 

\begin{lemma}[Pool consistency]
	%	Consider a client selection protocol  in Algorithm \ref{alg:protocol}. 
	For any client $i \in \calU$, and any honest clients $j_1, j_2$, we have,
	$$\Pr\left[		
	\begin{array}{l|l}
	\multirow{2}{*}{$\exists \ \omega_{j_1}^{(i)},\omega_{j_2}^{(i)}$} &
	\PoolVerify(\state_{j_1},\omega_{j_1}^{(i)}) = 1\land\\
	& \PoolVerify(\state_{j_2},\omega_{j_2}^{(i)}) = 0
	\end{array}
	\right] \le e^{-\Omega(\kappa)}$$
	\label{lemma:consistency}
\end{lemma}

%\vspace{-0.08in}
%\tn{It does not seem to match the definition?}
%\begin{proof}[Proof idea]
We omit the proof due to the space limit and outline the main intuition.
	As all the honest clients have the same view on the blockchain, the valid proofs $\omega_{j_1}^{(i)},\omega_{j_2}^{(i)}$ must have the same Merkle tree roots $\Merkleroot_t$ and $\Merkleroot_f$. Recall that, the client $i$ is considered to be selected if it is included in $\Merkleroot_t$ and $\Merkleroot_f$.
	Thus, the honest clients have the same view on whether or not the client $i$ is selected.

Next, we prove that the fraction of honest selected clients is proportional to the fraction of honest clients. Intuitively, the VRFs guarantee the randomness in selecting the qualified clients, i.e., the fraction of honest qualified clients is proportional to the fraction of honest clients. Plus, the dispute ensures that all honest qualified clients are selected. 
\begin{lemma} [Pool quality]
%	Consider a client selection protocol that are similar to the protocol in Algorithm \ref{alg:protocol}. 
%	Let $n$ be the number of clients,  $\alpha$ be the fraction of honest clients, Let $\gamma$ be the fraction of honest selected clients,
	Let $\calH$ be the set of honest clients in the set of selected clients $\calP$.
	For  $\epsilon > 0$, we have,
	\begin{align*}
		\Pr[\frac{\calH \cap \calP}{\calP} \ge \alpha(1-\epsilon)] \ge 1 - e^{-\Omega(nc - \log \kappa) }
	\end{align*}
	\label{lemma:pool0}
	where $n$ is the number of clients, $\alpha =  1 - \beta$ is the fraction of honest clients, and $c$ is the selection probability.
%	For any $\epsilon > 0$, with probability $1 - e^{-\Omega(n\mu)}$, the fraction of honest selected clients is at least $\alpha(1-\epsilon)$.
%\vspace{-0.2in}
\end{lemma}
%\vspace{-0.08in}
\begin{proof}
%	Let $\mu$ be the probability that the output of an VRF is smaller than $c$. 
	
%	\mt{what is $\alpha$?}
	
	Let $\calP' \supseteq	\calP$ be the set of qualified clients, i.e., the clients having VRF outputs smaller than $c2^\kappa$. 
	Let $\calH'$ and $\mathcal{M}'$ be the set of honest and colluding clients in $\calP'$, respectively.
%	Let $s_h$ be the number honest clients in $\calP$ and $s_m$ be the number of colluding clients in $\calP$. 
	
	We prove by bounding the number of qualified colluding clients. By the chain quality property of the blockchain \cite{garay2015bitcoin}, the adversary can create   at most $\kappa$ blocks among the last blocks used for creating the randomness. Thus, it has at most $\kappa$ randomness values to choose from. Using the Chernoff bound and union bound, for any $\epsilon' > 0$, we have
	\begin{equation*}
		\Pr[|\mathcal{M}'| \ge (1+\epsilon') n (1-\alpha) c] \le  \kappa e^{\Omega(nc)} = e^{-\Omega(nc - \log \kappa)}
	\end{equation*}
	For the honest clients, since the server cannot predict the outputs of the VRFs, thus, changing the randomness will not affect the probability that honest clients are selected.
	Using the Chernoff bound, for any $\epsilon' > 0$, we have,
	\begin{align*}
		\Pr[|\calH'| \le (1-\epsilon') n \alpha c] \le e^{-\Omega(nc)}
	\end{align*}
	
	By choosing $\epsilon'$ such that $\epsilon = \frac{1-\epsilon'}{1+\epsilon'}$, we have
	\begin{align*}
		\Pr\left[\frac{|\calH'|}{|\calP'|} \le \alpha(1-\epsilon)\right]
%		 &\le \Pr[|\calH'| \le (1-\epsilon') n \alpha c] \\ & \quad \quad +  	\Pr[|\mathcal{M}'| \ge (1+\epsilon') n (1-\alpha) c] \\
		&\le  e^{-\Omega(nc - \log \kappa)}
	\end{align*}
	
%	As each client has one chance to compute the VRF, the probability that a client is selected is $\mu$. Let $s_h$ be the number of honest 
	
%	Note that, based on the dispute, all honest clients in $S$ will be selected in the train round
%	Next, we will prove that all honest clients in $S$ will be selected in the train rounds. For each honest client $i$, we consider the two following cases:
%	\begin{itemize}
%		\item \emph{Case 1:} The server includes the client $i$ in the pool commitment. In this case, the client $i$ is selected.
%		\item \emph{Case 2:} The server does not include the client $i$ in the pool commitment. In this case, the server cannot provide the client $i$ a proof that the public key $\pk_i$ is included in the pool commitment. Thus, in the dispute resolution step, the client $i$ will submit a transaction to the blockchain. Then, the client $i$ is selected in the additional pool. 
%	\end{itemize}
%	Let $S$ be the set of clients in which the output of VRF is smaller than $c$. 
%	Using Chernoff bound, we have, with probability $1 - e^{-\Omega(n)}$, the fraction of honest clients in $S$ is at least $\alpha(1-\epsilon)$.
	
Recall that, the honest qualified clients are included either in the initial selection transaction or the final selection transaction (through dispute).
Thus, we have $\calH = \calH'$. Further, the selected clients must be qualified, i.e., $|\calP| \le |\calP'|$. Hence, we have $\frac{|\calH'|}{|\calP'|} \ge \frac{|\calH|}{|\calP|}$. Therefore, 
%	Based on the dispute, all honest clients in $\calP$ will be selected in the train round. Plus the selected clients must provide a proof that the output of their VRFs is smaller than $c$. Thus, 
%	the set of selected clients is a subset of $\calP$. Plus, all honest clients in $\calP$ is selected.
%	 since they can submit a transaction to the blockchain.
%	  Hence, we have,
	  \begin{align*}
	  \Pr\left[\frac{|\calH|}{|\calP|} \le \alpha(1-\epsilon)\right]
\le 	\Pr\left[\frac{|\calH'|}{|\calP'|} \le \alpha(1-\epsilon)\right] \le  e^{-\Omega(nc - \log \kappa)}
	  \end{align*}
	  
%	  with probability $1 - e^{-\Omega(n)}$, the fraction of honest selected clients is at least $\alpha(1-\epsilon)$.
\end{proof}
\vspace{-0.08in}

Finally, we show that even when the server can choose among up to $\kappa$ different randomess values, it has little chance to select a target client.

\begin{lemma}[Anti-targeting]
%	Consider a client selection protocol  in Algorithm \ref{alg:protocol}. 
Considering an honest client $i$,
%  Let $\pool$ be the set of selected clients.
we have
\begin{equation*}
	|\Pr[i \in \calP] - c| = e^{-\Omega(\kappa)},
\end{equation*}
where $c$ is the selection probability.
\label{lemma:targeting}
%\vspace{-0.25in}
\end{lemma}
\begin{proof}
	As we have shown in the proof of Lemma \ref{lemma:pool0}, an honest client is selected if the output of its VRF is smaller than a threshold. As the adversary cannot predict the output of the VRF of the client $i$, for any randomness $\rand$, the outputs of the VRF of $i$ cannot be distinguished with a random number.
	Thus, with probability $c$, the client $i$ is qualified. If the blockchain is secure (with probability $1- e^{-\Omega({\kappa})}$), the qualified client $i$ is selected. Thus, the probability that the client $i$ is selected is at most $c + e^{-\Omega(\kappa)}$ and at least $c - e^{-\Omega(\kappa)}$.
% 	the client $i$ is selected with probability $c$.
%	 and $j$, for any randomness $\rand$, the outputs of the VRFs of $i$ and $j$ can not be distinguished with a random number. Thus, the probabilities, that $i$ and $j$ are selected, are the same. 
\end{proof}

Together, lemmas \ref{lemma:consistency}, \ref{lemma:pool0}, \ref{lemma:targeting}  yield the security proof of our protocol.
\begin{theorem}[Secure pool selection]
	The pool selection protocol in Algorithm \ref{alg:protocol} achieves pool quality, pool consistency, and anti-targeting properties. 
\end{theorem}

%% file: performance.tex
\section{Experiments}\label{sec:exp}
We evaluate the performance of our protocol and the (insecure) baseline protocol (section~\ref{sec:attack}). 
Further, we analyze the dispute cost associated with the server and the clients.

\textbf{Setup.} We assume a public blockchain, e.g., Avalanche or Solana, with Solidity smart contracts. We use the VRF in the Libsodium cryptographic library\footnote{\url{https://github.com/algorand/libsodium/tree/draft-irtf-cfrg-vrf-03	
}} and the VRF verification in Solidity at \url{https://github.com/witnet/vrf-solidity}. 

We conducted all experiments on a CentOS machine Intel(R) Xeon(R) CPU E7-8894 v4 2.40GHz. 
We report the performance of the protocols in terms of blockchain storage cost (in KB), blockchain computation cost (in gas), and CPU time for the server and clients.

%compare the performance of our protocol with the baseline protocol. 
We choose the number of clients among  $10k$, $100k$, and $1000k$ and set the selection probability  $c = 1\%$ of that. The number of  training rounds in the FL process is assumed to be $1,000$. 
%Here, we consider a federated learning with $100,000$ clients. The probability that a client is selected is $1/100$.

%\vspace{-0.1in}
\subsection{Performance} 
%\vspace{-0.05in}

%\noindent \emph{Storage and computation on blockchain.}  

%measure the storage and computation cost on the blockchain for registration and $1,000$ training rounds.
%of the baseline protocol and our protocol. 
\begin{figure}[htp!]
%    \vspace{-0.1in}
	\centering
	\begin{subfigure}{0.4\textwidth}
		\includegraphics[width=\linewidth]{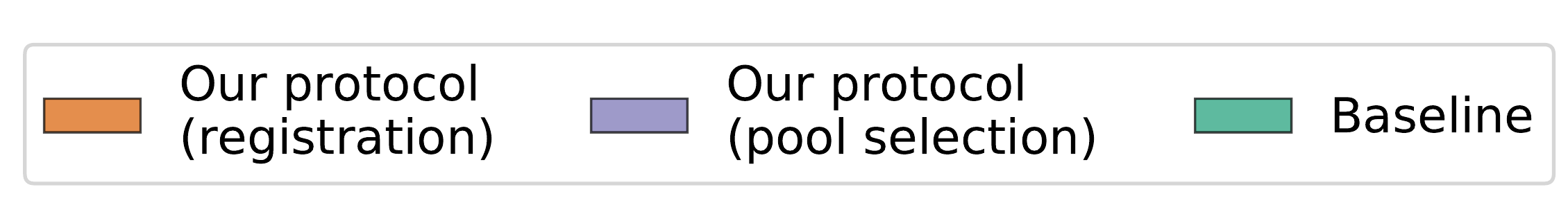}
	\end{subfigure}

	\begin{subfigure}{0.24\textwidth}
		\includegraphics[width=\linewidth]{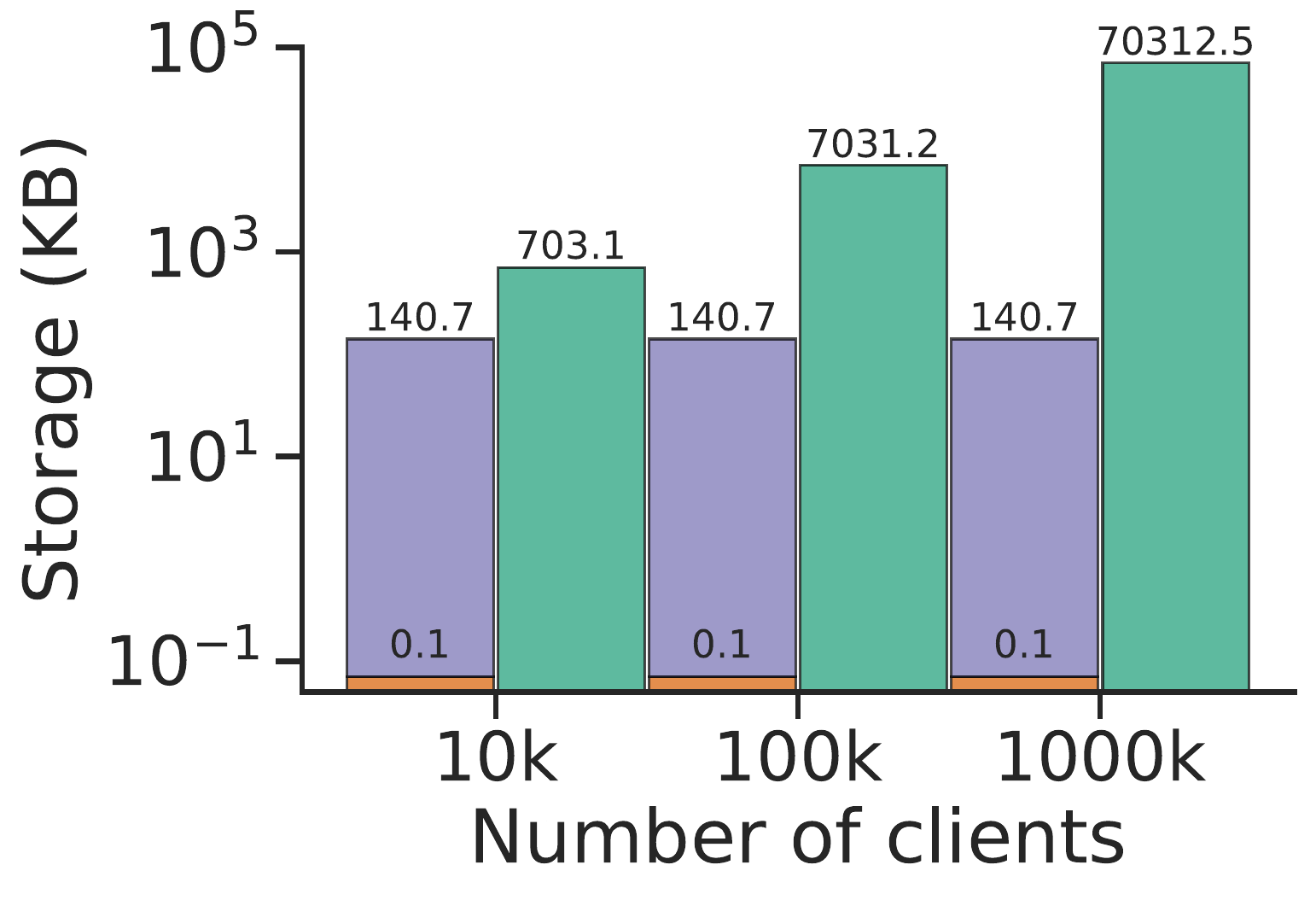}
		\caption{Storage \label{storage-reg}}
	\end{subfigure}
	\begin{subfigure}{0.24\textwidth}
		\includegraphics[width=\linewidth]{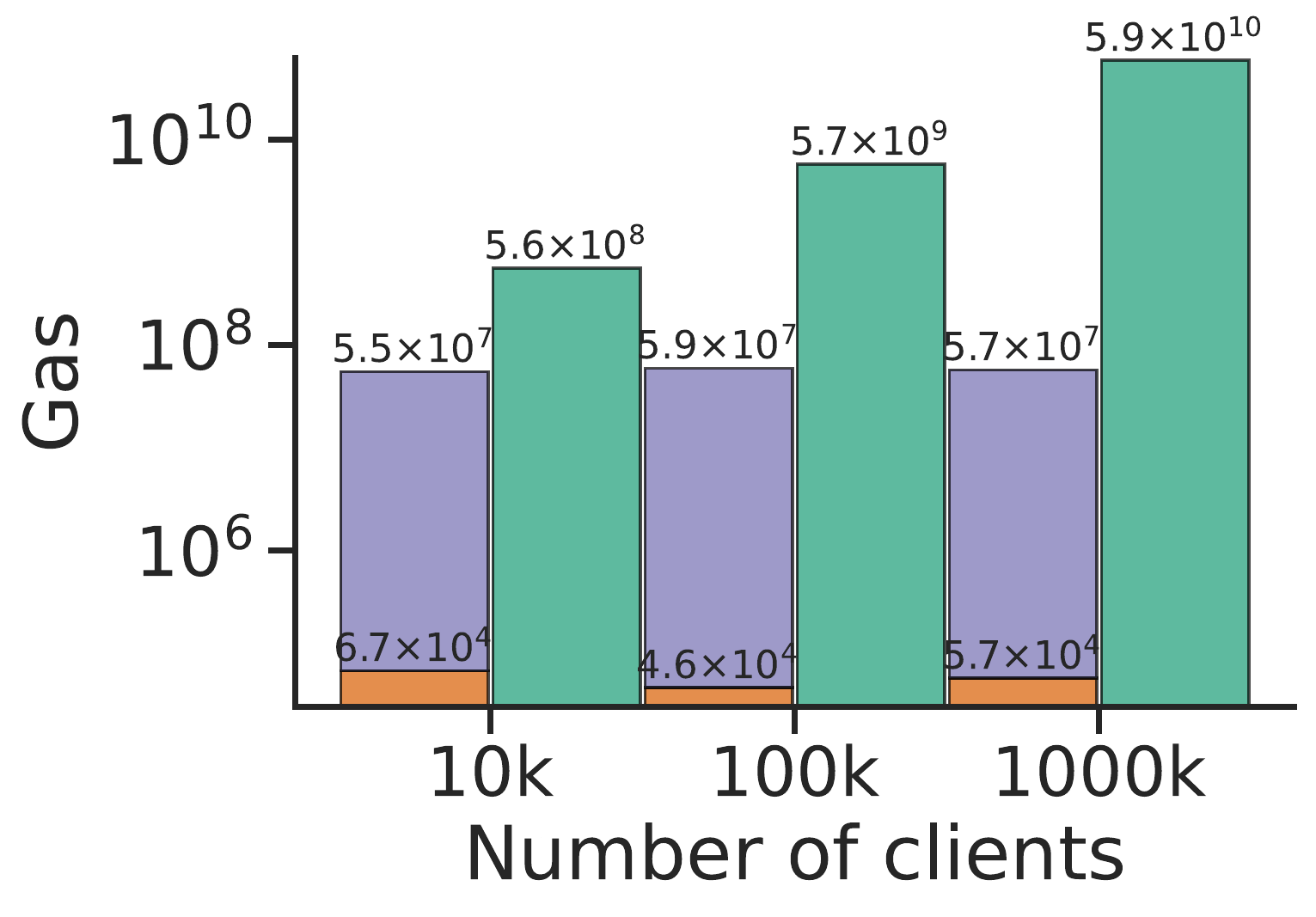}
		\caption{Computation cost (gas) \label{gas-reg}}
	\end{subfigure}
%	\subfloat[]{\includegraphics[width=0.48\linewidth]{figures/legend.pdf}}
%	\subfloat[][Storage. ]{\includegraphics[width=0.48\linewidth]{figures/storage-reg.pdf}\label{storage-reg}}
%	\hfil
%	\subfloat[][Computation cost (gas) ]{\includegraphics[width=0.48\linewidth]{figures/gas-reg.pdf}\label{gas-reg}}
	\caption{The storage  and computation costs on the blockchain for the registration and the pool selection in $1,000$ training rounds.}
	\label{fig:reg}
% \vspace{-0.2in}
\end{figure}

{
\noindent \emph{Storage and computation on blockchain.} 
As shown in Fig. \ref{fig:reg}, the storage and computation costs on blockchain for our protocol is significantly lower than those in the baseline protocol. Further, both the storage cost and the computation cost of our protocol remain constants when the number of clients increases. This contradicts the linear increase in the costs for the baseline protocol.  For example, when the number of clients $n = 1000k$, the total size of transactions in the baseline protocol and in our protocol are $70.3 MB$ and $0.1 KB$, respectively. Similarly, the total gas cost is $5.9\cdot 10^{10}$ in the baseline protocol and $5.7\cdot 10^7$ in our protocol. Thus, our protocol is up to several orders of magnitude more efficient than the (insecure) baseline protocol.

%\footnote{We use Avalanche as an example to show the transaction fee of our protocol an the baseline protocol. The estimated fee of the baseline protocol and our protocol are $1475$ AVAX ($\$110625$) and $1.5$ AVAX ($\$112.5$). 
%	Note that, we can replace Avalanche with other cryptocurrencies that allow Solidity as a smart contract language.

%in the baseline protocol, the total size of transactions is $7031.25 KB$, the total gas cost is $5.9\cdot 10^8$. In our protocol the 

%Each client in the baseline directly submits a transaction with its public key to the blockchain.
%In our protocol, to reduce the storage and computation on the blockchain, 
%%the server need to do some extra computation 
%the server gathers the public keys of the clients, constructs a Merkle tree, and then submits a transaction to the blockchain. 
%Thus, 
\noindent \emph{CPU time.} Only short computation times are needed for both the server and the clients in our protocol. For $n=1000k$ clients, the server in our protocol takes a $7.5s$ time to construct the Merkle tree during the registration, and a $5.8s$ time per round to verify the VRF proofs of the clients. All the computation is done using a single core. We note that this time can be reduced by several folds using parallel computing (not shown in here).
The computation for each client is also very short with  a negligible time in the registration, and a $0.031s$ time per training round.  

The baseline protocol incurs negligible computing times for both the server and the clients.

%\emph{Communication cost} The server and each client sends and receives $1.1GB$ and $1.13KB$ of data, respectively.

% and the server sends and receives $110.4 MB$ of data. 
%Compared with the baseline protocol, the server in our protocol needs to do some extra computation, i.e., constructing a Merkle tree before submitting a registration transaction. For the number of clients $n = 100k$, it takes $0.75s$ for the server to construct a Merkle tree. 
}
%This is not needed in the baseline protocol.
%Thus, the server need to perform some local computation 
%In contrast,  
%all clients in the baseline protocol submit transsactions that consist of their public key to the blockchain

%, and then submit a transaction of $72$ bytes to the blockchain. 

%As shown in Fig. \ref{fig:reg}, the storage and gas cost for the registration of our protocol on blockchain remains stable when the number of clients increases.
%On the other hand, as the number of clients increases, storage and gas cost for the registration of the baseline protocol also increases. When the number of clients $n = 100k$, the total size of transactions is $7031.25 KB$, the total gas cost is $5.9\cdot 10^8$.
%In the baseline protocol, each client submits a transaction containing a public key to the blockchain. The total size of the transaction is $7031.25 KB$.

%\mt{I suggest to remove this section B. In the journal version, we can add a new section: New Attacking Vector Analysis, which to discuss diff possible scenarios and prove that it's not possible. Here, focus on evaluate the cost, the performance, and re-affirm the theory, such as the model does satisfy those properties (defined earlier)... etc...}

%\vspace{-0.08in}
\subsection{Dispute cost}

We now measure the dispute cost of the server and the average dispute cost of each client. 
We consider a scenario in which the number of clients is $n = 1000k$, the selection probability $c = 1\%$, the probability that qualified clients submit a dispute transaction is $1\%$.
We report the cost of the server and the average cost of each client in $1,000$ training round.
% in which some clients submits dispute transactions.
%We break down the cost in our protocol into four categories: registration, initial selection, final selection, and dispute (see Fig \ref{fig:our}). 

%When the number of dispute clients increaes, 
%The cost of registration and    initial selection remain the same when the                                       
%in which an adversary that controls a subset of clients intentionally submits dispute transactions to degrade the protocol's performance. We compare the costs of attack, paid by the adversary, and the cost to response to the disputes, paid by the server.

\begin{table}[htp!]
	\centering
	%	\scriptsize
	\begin{tabular}{ c| rr}

		& Storage (KB) & Computation cost (gas) \\
				\hline
		The server & $70.35$ & $4.7 \times 10^9$\\
		Each client & $0.01$ & $2.1 \times 10^5$\\
%		\hline
	\end{tabular}
	\caption{The storage and computation costs for dispute. 
	}
	\label{tab:our}
\end{table}
As shown in Table. \ref{tab:our}, the average dispute cost of each client is much smaller than that, paid by the server. 
The storage costs for the server and each client are $140.7 KB$ and $0.01 KB$, respectively. Similarly, the gas cost for the server and each client are $4.7 \times 10^{9}$ and $2.1 \times 10^5$, respectively.

%% file: related.tex
\section{Related work}\label{sec:rel}
\noindent\textbf{Secure aggregation in FL.} Leveraging secret sharing and random masking, Bonawitz et al. \cite{bonawitz2017practical} propose a secure aggregation method and apply it to deep neural networks to aggregate client-provided model updates. In \cite{aono2017privacy} and \cite{zhang2020batchcrypt}, the authors utilize homomorphic encryption to blindly aggregate the model updates into global models. These secure aggregation protocols can scale up to millions of devices, and are robust to clients dropping out. Generic secure MPC based on secret sharing that securely computes any function among multiple parties \cite{damgaard2012multiparty,ben2019completeness,lindell2015efficient} can also be used as secure aggregation in FL. However, they are not scalable enough due to the high complexity in both computation and communication.  %the decryption key is distributed to the users, and the server uses homomorphic encryption to blindly aggregate the model updates. However, the authors assume that there is no collusion between the server and users so that the server cannot learn the decryption key. Therefore, the system does not work under our security model where users can be malicious and collude with the server.

Although these protocols provide strong security guarantees with respect to concealing the local model updates from the server, they are only applicable to an honest-but-curious adversary. They assume that the server honestly follows the protocol, including the random client selection. We show that the server can easily manipulate the selection process to bypass the secure aggregation and learn the local model update of a victim. We also devise a verifiable random selection protocol as a countermeasure to prevent the server from manipulating the selection of participating clients, thereby maintaining the security guarantees of secure aggregation protocols.

\smallskip
\noindent\textbf{Integration of Blockchain and FL.} Recently, there have been multiple studies focusing on integrating the immutability and transparency properties of blockchain into FL. For instance, Bao et al. \cite{bao2019flchain} propose FLChain which is an auditable and decentralized FL system that can reward the honest clients and detect the malicious ones. Zhang et al. \cite{zhang2020blockchain} propose a blockchain-based federated learning approach for IoT device failure detection. Kang et al. \cite{kang2020reliable} develop a reputation management scheme using blockchain to manage and select reliable clients, thereby avoiding unreliable model updates. In \cite{kim2019blockchained,ma2020federated}, the authors utilize blockchain for the exchange and aggregation of local model updates without a central server.

The above-mentioned systems cannot be employed directly to address the biased selection attack because they are not designed specifically for protecting client model updates. Additionally, they are not compatible to be used with a secure aggregation protocol. Our approach is different in a way that we use blockchain as a source of randomness for the client selection protocol, such that it enforces the random selection of clients, making the biased selection attack infeasible.

%% file: conclude.tex
% \vspace{-0.04in}
\section{Conclusion}\label{sec:con}
% \vspace{-0.04in}
In this paper, we have shown that using the secure aggregation protocols alone is not adequate to protect the local model updates from the server. Via our proposed biased selection attack, we have demonstrated that the server can manipulate the client selection process to learn the local model update of a victim, effectively circumventing the security guarantees of the secure aggregation protocols. To counter this attack and ensure privacy protection for the local model updates, we have proposed a verifiable client selection protocol using blockchain as a source of randomness. As a result, it enforces a random selection of clients in each training round, thereby preventing the server from manipulating the client selection process. We have proven its security against the proposed attack and analyzed its computation cost with Ethereum Solidity to show that it imposes negligible overhead on FL.